\documentclass{snapshotmfo}
\usepackage[utf8]{inputenc}
\usepackage[dvipsnames]{xcolor}
\usepackage{mathtools}
\usepackage{enumitem}
\usepackage{amsmath,amssymb}
\usepackage{amsthm}
\usepackage{tcolorbox}
\usepackage{epigraph}
\usepackage{hyperref}
\hypersetup{
    colorlinks=true,
    linkcolor=Aquamarine,
    filecolor=Aquamarine,      
    urlcolor=Aquamarine,
    citecolor=Aquamarine,
}

\usepackage[round]{natbib}
\usepackage[USenglish]{babel}
\usepackage{ellipsis}

\DeclareMathOperator\inter{int}

\newtheorem{theorem}{Theorem}[section]
\newtheorem{proposition}[theorem]{Proposition}
\newtheorem{lemma}[theorem]{Lemma}

\newtheorem{corollary}[theorem]{Corollary}

\newtheorem{definition}[theorem]{Definition}

\author{Mark Whitmeyer\thanks{Arizona State University. Email: \href{mailto:mark.whitmeyer@gmail.com}{mark.whitmeyer@gmail.com}. I thank Joseph Whitmeyer and Kun Zhang for their comments.}}
\title{\textcolor{ForestGreen}{Call the Dentist! A (Con-)Cavity in the Value of Information}}

\begin{document}
%%% Starting with version 1.29 the LaTeX package "bookmark" doesn't create the title bookmark any more.
% As a workaround, please uncomment and edit the following line:
%\pdfbookmark{Your title (without math mode, control sequences and special characters)}{snapshottitle}
% The second argument "snapshottitle" is just an identifier and can be left unchanged.

%%% Please insert your abstract here. 
\begin{abstract}
A natural way of quantifying the ``amount of information'' in decision problems yields a globally concave value for information. Another (in contrast, adversarial) way almost never does.
\end{abstract}
%%% If necessary, use the optional argument to supply an alternative text for the pdf metadata, cf. the comment on \title above: 
%\begin{abstract}[This is a plain text equivalent of your abstract used for the pdf metadata. It should give a brief overview of your snapshot. Please do not use more than 500 characters.]
%This is your abstract containing math mode, TeX control sequences, or some special characters. It should give a brief overview of your snapshot. Please do not use more than 500 characters.
%\end{abstract}

%%% Please insert the main body of your snapshot here.
\section{\textcolor{ForestGreen}{Introduction}}

We take another look at the classical nonconcavity of the value of information. This observation originates with \cite*{radner1984nonconcavity}, who point out that the value of information may not exhibit diminishing marginal returns. Since then, a number of papers have explored this idea further, including \cite*{Chade}, \cite*{de2007tight}, and \cite*{lara2020payoffs}.

An illustration as to why the value of information may not be concave  in its ``quantity,'' and in particular why the marginal value of information at quantity \(0\) may be zero goes something as follows. Take a decision problem with finitely many undominated actions. This means that the set of beliefs at which an action is optimal is a convex body: it is compact, convex, and possesses a non-empty interior. Accordingly, for a prior in the interior of one of these regions any information that does not move the prior much is worthless.

%we can find priors and non-zero ``quantities'' of information that do not change the optimality of a particular action for the decision-maker and are, hence, worthless.

However, it is difficult to quantify information and as noted by \cite*{radner1984nonconcavity}, ``the marginal productivity of information depends, of course, on the way the quantity of information is measured.'' They, as well as the aforementioned works, specify particular parametrizations of the quantity of information, which allows them to derive their results. The point of this paper is to argue that a natural way of quantifying information yields a value of information that is concave in its quantity.

Here is that quantification. Equating information with a Bayes-plausible distribution over posteriors \(F\) (given some prior), we specify the \textcolor{ForestGreen}{Quantity of Information} to be a convex function of the expected (convex) divergence of the posterior. Furthermore, as many different distributions over posteriors can yield the same quantity, we select distributions that are optimal for a decision-maker's (DM's) decision problem. That is all we need: in the main result of the paper we reveal that the value of information is concave in its amount, and so in non-trivial decision problems the marginal value of information is strictly positive at zero.

On the other hand, if we take the opposite approach, selecting distributions that are worst for a DM's decision problem, we discover that the value of information is almost never concave in its amount. In particular, the marginal value of information is almost always zero at zero. We finish the paper by departing from the expected-utility realm, and show that if the DM takes a max-min approach, and experiments are chosen benevolently, the value of information remains concave in its amount.

Essentially, these results follow from replication arguments. First, regardless of whether the DM is an expected-utility (EU) maximizer or a max-min (MEU) DM or whether information is chosen adversarially or benevolently, the quantity-of-information constraint binds: any optimal distribution must contain the maximal (or minimal, in the adversarial setting) ``amount'' of information. Second, given the first observation, the linearity of expectation means that an average of experiments always remains feasible for any average of information ``amounts.'' This produces the result (note that for an MEU DM, the impulse toward concavity is even stronger, as nature has less strategic freedom to choose her prior for a convex combination of ``amounts'').

\section{\textcolor{SeaGreen}{Model and Results}}

There is a single, Bayesian decision-maker (DM). There is an unknown state of the world \(\theta\) lying in some finite set of states \(\Theta\). \(\theta\) is distributed according to some full-support prior \(\mu \in \inter \Delta \left(\Theta\right)\). Information arrives according to a signal, stochastic map \(\pi \colon \Theta \to \Delta \left(S\right)\), where \(S\) is a compact set of signal realizations. Equivalently, information corresponds to a Bayes-plausible distribution \(F \in \mathcal{F}_\mu \subset \Delta^2\) over posteriors \(x \in \Delta\).\footnote{Given prior \(\mu\), the set of Bayes-plausible distributions, \(\mathcal{F}_\mu\), is \(\left\{F \in \Delta^2 \colon \mathbb{E}_F x = \mu\right\}\).}

The DM is an expected-utility maximizer with a compact set of actions, \(A\), available to her. She has a continuous utility function \(u \colon A \times \Theta \to \mathbb{R}\). For any posterior, we note the DM's value function,
\[V(x) \coloneqq \max_{a \in A} \mathbb{E}_{x}  u(a,\theta) \text{,}\] which we impose is not affine; \textit{viz.,} \(A\) contains at least two undominated actions.

Here is how we quantify the amount of information provided to the decision maker.
\begin{definition}
    Given function \(c\colon \Delta^2 \to \mathbb{R}\) satisfying i. \(c(x,\mu)\) is strictly convex, and ii. \(c(\mu,\mu) = 0\) for all \(\mu \in \Delta\); and (weakly) convex, strictly increasing \(\phi \colon \mathbb{R} \to \mathbb{R}\) satisfying \(\phi\left(0\right) = 0\), the \textcolor{ForestGreen}{Amount of Information} contained in the distribution over posteriors \(F\) is
\[D\left(F\right) \coloneqq \phi\left[\int_{\Delta} c\left(x,\mu\right)dF\left(x\right)\right]\text{.}\]
\end{definition}
Function \(c\) is a convex \textcolor{ForestGreen}{Divergence}, which are of central importance to models of flexible costly information acquisition (\cite*{bloedel2020cost}). One such \(c\) is the difference in Shannon entropy between the prior \(\mu\) and the posterior \(x\). An especially compelling justification of this definition of the ``amount of information,'' is \cite*{mensch2018cardinal} (for an affine \(\phi\)). In particular, his Theorem 1 reveals that this definition is the unique representation of preferences over experiments satisfying certain axioms. \cite*{denti1} contains a similar result.

For \(\eta \in \left[0,\bar{\eta}\right]\) (where \(\bar{\eta} \in \mathbb{R}_{++}\)), we define function \(W\colon \mathbb{R} \to \mathbb{R}\) to be
\[\label{starprog}\tag{{\color{OrangeRed} \(\clubsuit\)}} W\left(\eta\right) \coloneqq \max_{F \in \mathcal{F}_{\mu}, \ D(F) \leq \eta} \mathbb{E}_{F}V(x) \text{.}\]
That is, \(W\) is the DM's value for the amount of information \(\eta\) evaluated at an optimal distribution over posteriors \(F\) of amount \(\eta\). We term \(W\) the \textcolor{ForestGreen}{Efficient Value of Information}.

Let \[\mathcal{F}_{\mu}^* \coloneqq \left\{F \in \mathcal{F}_\mu \colon \mathbb{E}_F V \geq \max_{F' \in \mathcal{F}_{\mu}} \mathbb{E}_{F'}V\right\}\text{.}\]
We impose that \[\bar{\eta} < \inf_{F \in \mathcal{F}_{\mu}^*} D(F)\text{.}\]
That is, \(\mathcal{F}_{\mu}^*\) is the set of Bayes-plausible distributions over posteriors that optimize \(\mathbb{E}_F V\). As full information optimizes \(\mathbb{E}_F V\), \(\mathcal{F}_{\mu}^*\) is nonempty. As \(V\) is not affine, \(\inf_{F \in \mathcal{F}_{\mu}^*} D(F) > 0\). Our imposition means that the amount of information provided to the DM is less than any amount that yields the maximal value in the decision problem. This restriction engenders Lemma \ref{bind1}, though the concavity of the efficient value of information holds with or without it.

We say that the amount-of-information constraint binds if \(D(F^*) = \eta\) for any solution \(F^*\) to Program \ref{starprog}.
\begin{lemma}\label{bind1}
    The amount-of-information constraint binds.
\end{lemma}
\begin{proof}
    Suppose for the sake of contraction not. Take an arbitrary purported optimum \(F\). Let \(\bar{F}\) be the (Bayes-plausible) distribution over posteriors supported on the vertices of the simplex (induced by a fully informative experiment). Define \(F_{\lambda} \coloneqq \lambda F + \left(1-\lambda\right) \bar{F}\) for \(\lambda \in \left[0,1\right]\). By the continuity of \(\phi\), \(D(F_\lambda)\) is continuous in \(\lambda\), so, by the intermediate-value theorem, there exists a \(\lambda^* \in \left(0,1\right)\) such that \(D(F_{\lambda^*}) = \eta\). Moreover, as \(\eta \leq \bar{\eta} < \mathbb{E}_{\bar{F}}\),
    \[\mathbb{E}_{F}V < \lambda^* \mathbb{E}_{F}V + \left(1-\lambda^*\right)\mathbb{E}_{\bar{F}}V =  \mathbb{E}_{\lambda^* F + \left(1-\lambda^*\right)\bar{F}}V\text{,}\]
    which contradicts the optimality of \(F\). \end{proof}

Now, our first result: the efficient value of information is concave.
\begin{proposition}\label{weakconv}
    The efficient value of information, \(W\), is weakly concave in \(\eta\), strictly so if \(\phi\) is strictly convex.
\end{proposition}
\begin{proof}
Take distinct \(\eta_1\) and \(\eta_2\) and let \(F_1\) and \(F_2\) be respective optimizers of Program \ref{starprog}. Thus, by Lemma \ref{bind1}, \(D\left(F_1\right) = \eta_1\) and \(D\left(F_2\right) = \eta_2\), and so (if \(\phi\) is strictly convex)
    \[\upsilon \eta_1 + \left(1-\upsilon\right) \eta_2 = \upsilon D\left(F_1\right) + \left(1-\upsilon\right)D\left(F_2\right) \underset{(>)}{\geq} D\left(\upsilon F_1 + \left(1-\upsilon\right) F_2\right)\text{.}\] 
    Thus, \(\upsilon F_1 + \left(1-\upsilon\right) F_2\) is always feasible for information amount \(\upsilon \eta_1 + \left(1-\upsilon\right) \eta_2 \), strictly so if \(\phi\) is strictly convex. Accordingly, (if \(\phi\) is strictly convex)
    \[W(\upsilon \eta_1 + \left(1-\upsilon\right) \eta_2) \underset{(>)}{\geq} \mathbb{E}_{\upsilon F_1 + \left(1-\upsilon\right) F_2} V = \upsilon \mathbb{E}_{F_1}V + \left(1-\upsilon\right)\mathbb{E}_{F_2}V\text{,}\]
    where the inequality is implied by the binding information-amount constraint and the equality by the linearity of expectation. \end{proof}
    As there are at least two undominated actions,
    \begin{corollary}
        The marginal value of information at \(0\) is strictly positive.
    \end{corollary}
    Our next result provides a sufficient condition for \(W\) to be strictly concave when \(\phi\) is linear.
    \begin{proposition}\label{weakconvii}
    If there are just two undominated actions, the efficient value of information, \(W\), is strictly concave.
\end{proposition}
\begin{proof}
As there are just two actions, for any \(\eta \in \left[0,\bar{\eta}\right]\) any optimizer must be binary. Accordingly, for any distinct \(\eta_1\) and \(\eta_2\) with respective optimizers of Program \ref{starprog} \(F_1\) and \(F_2\), \(\nu F_1 + (1-\nu) F_2\) is strictly sub-optimal for information amount \(\nu \eta_1 + (1-\nu) \eta_2\).
\end{proof}
It is straightforward to construct a decision problem with three actions that is such that \(W\) is linear over a sub-interval of \(\left[0,\bar{\eta}\right]\).

An alternative formulation of the value of information almost never produces a concave value of information. In contrast to \(W\), which is generated by maximally efficient information acquisition, our next formulation is one in which information is acquired in a maximally inefficient manner. For \(\eta \in \left[0,\bar{\eta}\right]\) (where \(\bar{\eta} \in \mathbb{R}_{++}\)), we define function \(U\colon \mathbb{R} \to \mathbb{R}\) to be
\[\label{inefprog}\tag{{\color{OrangeRed} \(\heartsuit\)}} U\left(\eta\right) \coloneqq \min_{F \in \mathcal{F}_{\mu}, \ D(F) \geq \eta} \mathbb{E}_{F}V(x) \text{.}\]
We term \(U\) the \textcolor{ForestGreen}{Inefficient Value of Information}.

We say a value of information is \textcolor{ForestGreen}{Almost Never Concave} if the set of priors \(\mu\) at which the value of information is \(0\) for all \(\eta \in \left[0,\hat{\eta}_{\mu}\right]\) (\(\hat{\eta}_\mu > 0\)) is dense in \(\Delta\).
\begin{proposition}
    Let \(\Theta\) and the number of undominated actions be finite. Then the inefficient value of information is almost never concave.
\end{proposition}
\begin{proof}
    As the number of undominated actions is finite \(V\) is the maximum of a finite number of affine functions. We may project \(V\) onto \(\Delta\), yielding a finite collection polytopes \(C\). \(\S2.1\) of \cite*{flexibilitypaper} contains more information about this object. On each element \(C_i \in C\), \(V\) is affine. Moreover, by construction \(\mathcal{C} \coloneqq \cup_{i = 1}^{m} \inter C_i\) is dense in \(\Delta\) (where \(m\) is the number of undominated actions in \(A\)). 

    %The set of points of \(\Delta\) where \(V\) is differentiable is \(\mathcal{C} \coloneqq \cup_{i = 1}^{m} \inter C_i\) (where \(m\) is the number of actions in \(A\) that are not weakly dominated). By Theorem 25.5 in \cite*{rockafellar2015convex} (p.246), \(\mathcal{C}\) is a dense subset of \(\inter \Delta\) and, hence, of \(\Delta\).

    Observe that for any \(C_i\) and any prior \(\mu \in \inter C_i\), the DM is indifferent between no information and any (Bayes-plausible) distribution over posteriors supported on \(C_i\). Evidently, for any \(C_i\), any \(\mu \in \inter C_i\), and any Bayes-plausible distribution \(G\) supported on the extreme points of \(C_i\), \(D(G) > 0\). Consequently, for any \(\mu \in \inter C_i\), there exists a cost threshold \(\hat{\eta}_\mu > 0\) such that for all \(\eta \in \left[0,\hat{\eta}_\mu\right]\), any solution to Program \ref{inefprog} is supported on a subset of \(C_i\). Thus, for any \(\eta \in \left[0,\hat{\eta}_\mu\right]\), \(U(\eta) = U(0)\). \end{proof}

\section{A Max-min DM}

Now, suppose that the DM is not an expected-utility maximizer, but instead evaluates decisions according to a max-min criterion \`{a} la \cite{gilboa1989maxmin}. Let \(A\) be finite. There is a compact and convex subset of feasible priors \(M \subseteq \Delta\), which we specify is of full-dimension in \(\Delta\). Following \cite{ccelen2012informativeness}, the value of experiment \(\pi\) to the DM is 
\[T\left(\pi\right) \coloneqq \max_{\sigma} \min_{\mu \in M} \mathbb{E}_{B(\mu, \pi)}\mathbb{E}_x \mathbb{E}_{\sigma}u(a,\theta)\text{,}\]
where \(\sigma \colon S \to \Delta\left(A\right)\) is the DM's signal-dependent choice of action, and \(B\left(\mu,\pi\right) \in \Delta^2\) is the Bayes-map that takes as input prior \(\mu\) and statistical experiment \(\pi\) and outputs a distribution over posteriors \(B\left(\mu,\pi\right)\).\footnote{Please refer to, e.g., \cite*{denti1} or \cite*{denti2}.}

\begin{definition}
    Given function \(c\colon \Delta^2 \to \mathbb{R}\) satisfying i. \(c(x,\mu)\) is strictly convex, and ii. \(c(\mu,\mu) = 0\) for all \(\mu \in \Delta\); element \(\tilde{\mu} \in \inter \Delta\); and (weakly) convex, strictly increasing \(\phi \colon \mathbb{R} \to \mathbb{R}\) satisfying \(\phi\left(0\right) = 0\), the \textcolor{ForestGreen}{Amount of Information} contained in experiment \(\pi\) is
\[C\left(\pi\right) \coloneqq \phi\left[\int_{\Delta} c\left(x,\tilde{\mu}\right)dB\left[\tilde{\mu},\pi\right]\left(x\right)\right]\text{.}\]
\end{definition}

Now, for \(\eta \in \left[0,\bar{\eta}\right]\) (where \(\bar{\eta} \in \mathbb{R}_{++}\)), we define function \(\bar{W}\colon \mathbb{R} \to \mathbb{R}\) to be
\[\label{bigstar}\tag{\textcolor{OrangeRed}{\(\spadesuit\)}}\bar{W}\left(\eta\right) \coloneqq \max_{\pi, \ C(\pi) \leq \eta} T(\pi) \text{.}\]
We term \(\bar{W}\) the \textcolor{ForestGreen}{Max-min Value of Information}. As we did before, we bound the informativeness of \(\pi\): let\[\Pi^* \coloneqq \left\{\pi \colon T(\pi) = T(\bar{\pi})\right\}\text{,}\] where \(\bar{\pi}\) is the fully informative experiment. We impose that \[\bar{\eta} < \inf_{\pi \in \Pi^*} C(\pi)\text{.}\]

    \begin{lemma}\label{bind2}
    The amount-of-information constraint binds.
\end{lemma}
\begin{proof}
Mirroring the proof to Lemma \ref{bind1}, we suppose for the sake of contradiction that the constraint doesn't bind. As noted there, given a purported optimal \(\pi\), there exists some \(\lambda \in \left(0,1\right)\) such that \(\lambda \pi + (1-\lambda)\bar{\pi}\) is feasible, where (recall) \(\bar{\pi}\) is the fully informative experiment.

Let \(\sigma\) be an equilibrium strategy when the experiment is \(\pi\); and given \(\sigma\), let \(\mu^*\) be an arbitrary solution to 
\[\min_{\mu \in M} \mathbb{E}_{B(\mu, \pi)}\mathbb{E}_x \mathbb{E}_{\sigma}u(a,\theta)\text{.}\]

Next, let \(\bar{\sigma}\) be an action strategy when the experiment is \(\bar{\pi}\) that picks an arbitrary optimal action in every state. By construction, the DM's payoff, \(T\left(\lambda \pi + (1-\lambda)\bar{\pi}\right)\), is weakly greater than
\[\label{heartprog}\tag{\textcolor{OrangeRed}{\(\diamondsuit\)}}\min_{\mu \in M} \left\{\lambda \mathbb{E}_{B(\mu, \pi)}\mathbb{E}_x \mathbb{E}_{\sigma}u(a,\theta) + (1-\lambda) \mathbb{E}_{B(\mu, \bar{\pi})}\mathbb{E}_x \mathbb{E}_{\bar{\sigma}}u(a,\theta)\right\}\text{.}\]
Evidently, for an arbitrary solution to Program \ref{heartprog}, \(\bar{\mu}\), we have
\[\mathbb{E}_{B(\bar{\mu}, \pi)}\mathbb{E}_x \mathbb{E}_{\sigma}u(a,\theta) \geq \mathbb{E}_{B(\mu^*, \pi)}\mathbb{E}_x \mathbb{E}_{\sigma}u(a,\theta)\text{.}\]
Moreover, as \(\pi^{\max}\) is fully informative,
\[\mathbb{E}_{B(\bar{\mu}, \bar{\pi})}\mathbb{E}_x \mathbb{E}_{\bar{\sigma}}u(a,\theta) > \mathbb{E}_{B(\bar{\mu}, \pi)}\mathbb{E}_x \mathbb{E}_{\sigma}u(a,\theta)\text{.}\]
Combining these, we conclude that \(\pi\) is strictly sub-optimal, a contradiction.\end{proof}
Here is our last result: the max-min value of information is concave.
\begin{proposition}
    The max-min value of information, \(\bar{W}\), is weakly concave in \(\eta\), strictly so if \(\phi\) is strictly convex.
\end{proposition}
\begin{proof}
Take distinct \(\eta_1\) and \(\eta_2\) and let \(\pi_1\) and \(pi_2\) be respective optimizers of Program \ref{bigstar}. By the same logic as that in the proof for Proposition \ref{weakconv}, \(\upsilon \pi_1 + \left(1-\upsilon\right) \pi_2\) is always feasible for information amount \(\upsilon \eta_1 + \left(1-\upsilon\right) \eta_2 \), strictly so if \(\phi\) is strictly convex.

Let \(\sigma_i\) (\(i = 1,2\)) be equilibrium strategies for experiments \(\pi_1\) and \(\pi_2\), respectively. Obviously,
\[\label{eq1}\tag{\textcolor{OrangeRed}{\(\bigstar\)}}\min_{\mu \in M} \left\{\eta \mathbb{E}_{B(\mu, \pi_1)}\mathbb{E}_x \mathbb{E}_{\sigma_1}u(a,\theta) + (1-\eta) \mathbb{E}_{B(\mu, \pi_2)}\mathbb{E}_x \mathbb{E}_{\sigma_2}u(a,\theta)\right\}\]
is weakly greater than
\[\eta \min_{\mu \in M} \mathbb{E}_{B(\mu, \pi_1)}\mathbb{E}_x \mathbb{E}_{\sigma_1}u(a,\theta) +  (1-\eta) \min_{\mu \in M}\mathbb{E}_{B(\mu, \pi_2)}\mathbb{E}_x \mathbb{E}_{\sigma_2}u(a,\theta)\text{,}\]
and Expression \ref{eq1}, itself, is weakly (strictly, if \(\phi\) is strictly convex) less than \(\bar{W}(\upsilon \eta_1 + \left(1-\upsilon\right) \eta_2)\). Thus, we have (strict) concavity. \end{proof}

\bibliography{sample.bib}

\end{document}